\RequirePackage{amsmath}
\documentclass[a4paper,12pt]{article}
\usepackage{graphicx}
\usepackage{amsmath}
\usepackage{amssymb}
\usepackage{amsthm}

\usepackage{algorithmic}
\usepackage{algorithm}

\usepackage{array}
\newcolumntype{C}{>{$\displaystyle}c<{$}}

\newcommand{\seq}[2]{\frac{\strut#1}{\strut#2}}
\usepackage{enumitem}

\usepackage{mathtools}

\usepackage{latexsym,mathrsfs}
\usepackage{graphics}
\usepackage{url}
\usepackage{float}

\usepackage{algorithmic}
\usepackage{algorithm}

\usepackage{tikz}
\usetikzlibrary{arrows,automata,positioning}

\usepackage{booktabs}
\usepackage{tabularx}

\usepackage{color}

\usepackage[affil-it]{authblk}

\newtheorem{defn}{Definition}
\newtheorem{lem}{Lemma}

\newtheorem{thm}{Theorem}

\newtheorem{prop}{Proposition}

\begin{document}


\title{Temporal Logic for Social Networks}
\author{Vitor Machado%
  \thanks{Electronic address: \texttt{vmachado@cos.ufrj.br}}}
\affil{System Engineering and Computer Science Program \\ Federal University of Rio de Janeiro}

\author{Mario Benevides%
  \thanks{Electronic address: \texttt{mario@ic.uff.br}}}
\affil{Computing Institute \\ Fluminense Federal University}

\date{\today}

\maketitle

\begin{abstract}
This paper introduces a logic with a class of social network models that is based on standard Linear Temporal Logic (LTL), leveraging the power of existing model checkers for the analysis of social networks. We provide a short literature overview, and then define our logic and its axiomatization, present some simple motivational examples of both models and formulas, and show its soundness and completeness via a translation into propositional formulas. Lastly, we briefly discuss model checking and time complexity analysis.
\end{abstract}

\section{Introduction}
Social networks are a growing topic of study since the Internet social network services boom starting in the late 1990's with networks such as Friendster, MySpace, LinkedIn, and others, and more recently Facebook.
These services allow people from all over the world to interact in meaningful ways, and as such, have gained tremendous importance in popular culture in recent years. They also impose some challenges, such as the spreading of false information, usually referred to as ``fake news''.

In a broader sense, these networks are commonly defined by a set of entities (individuals, communities or societies), relations between pairs of entities, and a set of social interactions. For instance, in Facebook, an entity is a person, user of the service; the relations between pairs of entities are the friendship relations and the social interactions are posts which appear on the user's timeline.

Our goal in this paper is to develop a logic based on LTL, denoted LTL-SN for social network models without unadoption. Semantics are standard LTL, but we are interested in models with social network properties of adoption, where an agent adopts the behavior when the majority of its friends have the behavior.

The roadmap is as follows: In the following subsection we briefly go over some of the literature that inspired or served as a basis for this work, providing an overview of the past and current tendencies in the study social networks with such frameworks. In chapter \ref{language-and-semantics} we introduce the language and its semantics, present some simple examples of model evolution and formulas, and in chapter \ref{axioms-system} we present some important definitions that lead on to the axiomatization and finally its soundness and completeness proofs in chapter \ref{soundness-and-completeness}. In chapter \ref{model-checking} we briefly discuss model checking and its time complexity, and then present our closing words and future works in chapter \ref{conclusion-future-works}.

\subsection{Literature overview}
Papers \cite{ReachingConsensus,WhenRationalDisagreementImpossible}, seminal works in social networks theory, deal with how a group of agents should reach consensus by assessing the probabilities that each agent assigns to a given proposition, and that informed disagreements are therefore irrational behavior.



A threshold-model is presented in \cite{10.1007/978-3-642-24829-0_20,DBLP:journals/corr/abs-1708-01477}, in which nodes influenced by their neighbours can adopt one out of several alternatives. It focuses on establishing conditions for adoption of a ``product'' by a node, and whether the process has a unique outcome or not, or an unavoidable one. Nodes have neighbours and are always aware of their product adoptions. Nodes have a threshold $\theta$, and will adopt a product if more than $\theta$ of its neighbours have adopted or not. The paper focuses strictly on these dynamics and outcome conditions, and does not provide a logic formalism.

In \cite{DBLP:journals/corr/abs-1708-01477} a different approach to threshold models is taken by constructing the dynamics using action models and product update. The paper proposes an action model approach to the updates, which works by applying the product between the network graph and a graph that encodes the decision rules, and is uniformly followed by all agents.

The paper \cite{ZHEN2011275} investigates the phenomenon of peer pressure, by introducing a model where a person's preferences are changed in response to the preferences of a ``peer group''. Relationships among agents are defined with the ``facebook logic'' (symmetric relation of ``friendship'', $\sim_w$, in a state $w$).

Private and public opinions are considered under the dynamics of social influence in \cite{ReflectingSocialInfluenceNetworks}. It provides rules for ``strong'' and ``weak'' influence, and also rules for when an agent will express its true inner opinions or a different external opinion. Agents are capable of evolving their epistemic knowledge to infer the inner opinions of others, by discarding states which are contradictory to the rules of social influence.

In \cite{DynamicEpistemicLogicsDiffusionPredictionSocialNetworks} a ``minimal'' threshold-model based logic is presented, in the sense that it does not use static modalities or hybrid logic for its diffusion process. The authors go on to introduce an epistemic version of this logic, where agents must have sufficient information about the behaviors of their neighbours before adopting the behavior, as standard in epistemic logics. This framework serves as the main inspiration for the work we present here.

\section{Language and semantics}
\label{language-and-semantics}
In this section we present the language of LTL-SN, define certain properties of its models and its semantics.

\begin{defn}[Language]
The language of LTL-SN is given by the following BNF:
    \begin{equation*}
        \varphi \vcentcolon = \top \mid N_{ab} \mid \beta_a  \mid \neg \varphi \mid \varphi_1 \land \varphi_2 \mid X \varphi \mid \varphi_1 U \varphi_2
    \end{equation*}
\end{defn}
where $a, b \in \mathcal{A}$,

We use the standard abbreviations $\bot \equiv \neg \top$, $\varphi \lor \phi \equiv \neg ( \neg \varphi \land \neg \phi)$, $\varphi \rightarrow \phi \equiv \neg( \varphi \land \neg \phi)$, $F \varphi \equiv \top U \varphi$ and $G \varphi \equiv \neg  F \neg \varphi$.

\begin{defn}[Model]
\label{defn-model}
A model is a tuple ${\cal M} = ( \mathcal{A}, N , \theta, I)$, where
\begin{itemize}
    \item  $\mathcal{A}$ is a finite  set of agents;
    \item $N$ is a neighborhood function $N: \mathcal{A} \mapsto 2^{\mathcal{A}}$ such that $N$ is irreflexive, serial and symmetric;
    \item $\theta \in [0,1]$ is a uniform adoption threshold;
    \item $I$ is an initial behavior set;
\end{itemize}
\end{defn}

\begin{defn}[Relation $\leq$]
Let $\leq$ be the binary relation over $2^{\mathcal{A}} \times 2^{\mathcal{A}}$, where
\begin{equation*}
B \leq B' \  \text{iff} \  B' = B \cup \left\{ a \in \mathcal{A} : \frac{\lvert N(a) \cap B \rvert}{\lvert N(a) \rvert} > \theta \right\}
\end{equation*}
\end{defn}

\begin{defn}[Behaviors]
$\mathcal{B} \subseteq 2^{\mathcal{A}}$ is the set of behavior sets.
\end{defn}




\begin{defn}[Path on $\mathcal{M}$]
We define a path $b = b_0, b_1, b_2, \ldots$ on a model $\mathcal{M} = (\mathcal{A}, N , \theta, I)$ to be such that $b_0 = I$ and $b_i \leq b_{i+1}$.
\end{defn}

\begin{lem}
[Path uniqueness]
\label{path-uniqueness}
There is exactly one path $b$ such that $b_0 = I$.
\end{lem}

\begin{proof}
Follows directly from the fact that for any $b_i$ the relation $\leq$ defines strictly one $b_{i+1}$ such that $b_i \leq b_{i+1}$.
\end{proof}

Given lemma \ref{path-uniqueness}, we can define satisfaction for a model directly over paths.

\begin{defn}[Satisfaction]
\label{def-semantics}
Given a model $\mathcal{M} = (\mathcal{A}, N , \theta, I)$ and its path $b$. The notion of satisfaction of formulas in $\mathcal{M}$ at a position $i$ is defined as follows:
\[
        \begin{array}{lll}
    \mathcal{M}, i \models \beta_a &\text{iff} &a \in b_i \\
    \mathcal{M}, i \models N_{ax} &\text{iff} &x \in N_a \\
    \mathcal{M}, i \models \neg \varphi &\text{iff} &\mathcal{M}, i \not\models \varphi \\
    \mathcal{M}, i \models \varphi_1 \land \varphi_2 &\text{iff} &\mathcal{M}, i \models \varphi_1 \  \text{and} \  \mathcal{M}, i \models \varphi_2 \\
    \mathcal{M}, i \models X \varphi &\text{iff} &\mathcal{M}, i+1 \models \varphi \\
    \mathcal{M}, i \models \varphi U \psi &\text{iff} & \text{there exists a} \  n \geq 0 \  \text{such that} \\
        & & \mathcal{M}, i + n \models \psi \  \text{and} \  \mathcal{M}, i + j \models \varphi \  \text{for all} \  0 \leq j < n
\end{array}
\]    
\end{defn}

All LTL operators, plus propositional operators $N_{ax}$ and $\beta_a$ defined in a standard way.

\begin{defn}[Validity]
$\models \varphi$ iff $\mathcal{M}, i \models \varphi$ for all models $\mathcal{M}$ and positions $i$.
\end{defn}

\subsection{Examples}

In figure \ref{fig:example-domination} the behavior starts with agent $a$, and spreads through the entire network, while in figure \ref{fig:example-cluster} a slightly different network setting (agents $d$ and $e$ are related) creates a ``cluster'' that resists the behavior.

Notice the formula $(N_{bd} \land \beta_b) \lor (N_{fd} \land \beta_f) \rightarrow X \beta_d$ holds in both of them (though vacuously in figure \ref{fig:example-cluster}), and that it determines how agent $d$ is able to incorporate the behavior itself. A generalization of formulas like this one will be used in the axiomatization present in section \ref{axioms-system}.

For an example of an Until formula, consider $\neg (b_d \land b_e \land b_f) U b_d$. It holds in figure \ref{fig:example-domination}, and expresses the fact that agent $d$ will not adhere to the behavior until one of its friends does. Notice it does not hold in figure \ref{fig:example-cluster} because $d$ never gets to incorporate the behavior. Considering the previous statement, it is easy to see that, for instance, $G \neg \beta_d$ holds only in figure \ref{fig:example-cluster}.

    
\begin{figure}
    \centering
    \begin{tikzpicture}[align=center,semithick, every label/.append style={font=\scriptsize}]
        \tikzstyle{every state}=[draw=black,text=black]
        \tikzstyle{filled}=[fill=black!30]
        \node[state,inner sep=0pt,minimum size=15pt] (A) [filled, label=a]                                  {};
        \node[state,inner sep=0pt,minimum size=15pt] (B) [above right=0.12cm and 0.5cm of A, label=b]       {};
        \node[state,inner sep=0pt,minimum size=15pt] (C) [below right=0.12cm and 0.5cm of A, label=below:c] {};
        \node[state,inner sep=0pt,minimum size=15pt] (D) [right=0.12cm and 0.5cm of B, label=d]             {};
        \node[state,inner sep=0pt,minimum size=15pt] (E) [right=0.12cm and 0.5cm of C, label=below:e]       {};
        \node[state,inner sep=0pt,minimum size=15pt] (F) [below right=0.12cm and 0.5cm of D, label=f]       {};

        \path
            (A) edge                 node [] {} (C)
            (B) edge                 node [] {} (C)
            (B) edge                 node [] {} (D)
            (B) edge                 node [] {} (E)
            (B) edge [bend right=15] node [] {} (F)
            (C) edge                 node [] {} (E)
            (D) edge                 node [] {} (F)
            (E) edge                 node [] {} (F);
    \end{tikzpicture}
    \vspace*{15pt}
    \hspace*{15pt}
    \begin{tikzpicture}[align=center,semithick, every label/.append style={font=\scriptsize}]
        \tikzstyle{every state}=[draw=black,text=black]
        \tikzstyle{filled}=[fill=black!30]
        \node[state,inner sep=0pt,minimum size=15pt] (A) [filled, label=a]                                          {};
        \node[state,inner sep=0pt,minimum size=15pt] (B) [above right=0.12cm and 0.5cm of A, label=b]               {};
        \node[state,inner sep=0pt,minimum size=15pt] (C) [filled, below right=0.12cm and 0.5cm of A, label=below:c] {};
        \node[state,inner sep=0pt,minimum size=15pt] (D) [right=0.12cm and 0.5cm of B, label=d]                     {};
        \node[state,inner sep=0pt,minimum size=15pt] (E) [right=0.12cm and 0.5cm of C, label=below:e]               {};
        \node[state,inner sep=0pt,minimum size=15pt] (F) [below right=0.12cm and 0.5cm of D, label=f]               {};

        \path
            (A) edge                 node [] {} (C)
            (B) edge                 node [] {} (C)
            (B) edge                 node [] {} (D)
            (B) edge                 node [] {} (E)
            (B) edge [bend right=15] node [] {} (F)
            (C) edge                 node [] {} (E)
            (D) edge                 node [] {} (F)
            (E) edge                 node [] {} (F);
    \end{tikzpicture}
    \hspace*{15pt}
    \begin{tikzpicture}[align=center,semithick, every label/.append style={font=\scriptsize}]
        \tikzstyle{every state}=[draw=black,text=black]
        \tikzstyle{filled}=[fill=black!30]
        \node[state,inner sep=0pt,minimum size=15pt] (A) [filled, label=a]                                          {};
        \node[state,inner sep=0pt,minimum size=15pt] (B) [above right=0.12cm and 0.5cm of A, label=b]               {};
        \node[state,inner sep=0pt,minimum size=15pt] (C) [filled, below right=0.12cm and 0.5cm of A, label=below:c] {};
        \node[state,inner sep=0pt,minimum size=15pt] (D) [right=0.12cm and 0.5cm of B, label=d]                     {};
        \node[state,inner sep=0pt,minimum size=15pt] (E) [filled, right=0.12cm and 0.5cm of C, label=below:e]       {};
        \node[state,inner sep=0pt,minimum size=15pt] (F) [below right=0.12cm and 0.5cm of D, label=f]               {};

        \path
            (A) edge                 node [] {} (C)
            (B) edge                 node [] {} (C)
            (B) edge                 node [] {} (D)
            (B) edge                 node [] {} (E)
            (B) edge [bend right=15] node [] {} (F)
            (C) edge                 node [] {} (E)
            (D) edge                 node [] {} (F)
            (E) edge                 node [] {} (F);
    \end{tikzpicture}
    
    
    \begin{tikzpicture}[align=center,semithick, every label/.append style={font=\scriptsize}]
        \tikzstyle{every state}=[draw=black,text=black]
        \tikzstyle{filled}=[fill=black!30]
        \node[state,inner sep=0pt,minimum size=15pt] (A) [filled, label=a]                                          {};
        \node[state,inner sep=0pt,minimum size=15pt] (B) [filled, above right=0.12cm and 0.5cm of A, label=b]       {};
        \node[state,inner sep=0pt,minimum size=15pt] (C) [filled, below right=0.12cm and 0.5cm of A, label=below:c] {};
        \node[state,inner sep=0pt,minimum size=15pt] (D) [right=0.12cm and 0.5cm of B, label=d]                     {};
        \node[state,inner sep=0pt,minimum size=15pt] (E) [filled, right=0.12cm and 0.5cm of C, label=below:e]       {};
        \node[state,inner sep=0pt,minimum size=15pt] (F) [filled, below right=0.12cm and 0.5cm of D, label=f]       {};

        \path
            (A) edge                 node [] {} (C)
            (B) edge                 node [] {} (C)
            (B) edge                 node [] {} (D)
            (B) edge                 node [] {} (E)
            (B) edge [bend right=15] node [] {} (F)
            (C) edge                 node [] {} (E)
            (D) edge                 node [] {} (F)
            (E) edge                 node [] {} (F);
    \end{tikzpicture}
    \hspace*{15pt}
    \begin{tikzpicture}[align=center,semithick, every label/.append style={font=\scriptsize}]
        \tikzstyle{every state}=[draw=black,text=black]
        \tikzstyle{filled}=[fill=black!30]
        \node[state,inner sep=0pt,minimum size=15pt] (A) [filled, label=a]                                          {};
        \node[state,inner sep=0pt,minimum size=15pt] (B) [filled, above right=0.12cm and 0.5cm of A, label=b]       {};
        \node[state,inner sep=0pt,minimum size=15pt] (C) [filled, below right=0.12cm and 0.5cm of A, label=below:c] {};
        \node[state,inner sep=0pt,minimum size=15pt] (D) [filled, right=0.12cm and 0.5cm of B, label=d]             {};
        \node[state,inner sep=0pt,minimum size=15pt] (E) [filled, right=0.12cm and 0.5cm of C, label=below:e]       {};
        \node[state,inner sep=0pt,minimum size=15pt] (F) [filled, below right=0.12cm and 0.5cm of D, label=f]       {};

        \path
            (A) edge                 node [] {} (C)
            (B) edge                 node [] {} (C)
            (B) edge                 node [] {} (D)
            (B) edge                 node [] {} (E)
            (B) edge [bend right=15] node [] {} (F)
            (C) edge                 node [] {} (E)
            (D) edge                 node [] {} (F)
            (E) edge                 node [] {} (F);
    \end{tikzpicture}
\caption{This figure presents a social network (with $\theta = 1/3$) with some agents and the evolution of a behavior. Nodes are agents, and the edges represent the relation between them. Each figure, ordered from left to right and top to bottom, represents a distinct position as the network evolves. The behavior eventually dominates the network.}
\label{fig:example-domination}
\end{figure}
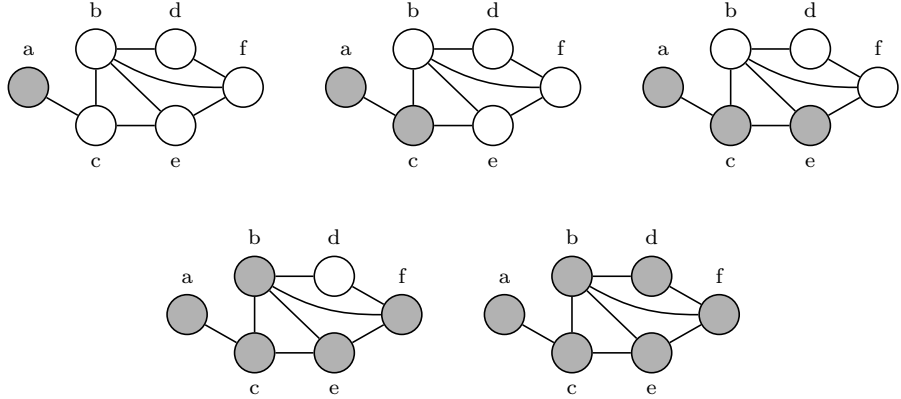


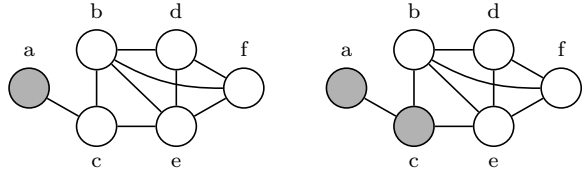
\begin{figure}
    \centering
    \begin{tikzpicture}[align=center,semithick, every label/.append style={font=\scriptsize}]
        \tikzstyle{every state}=[draw=black,text=black]
        \tikzstyle{filled}=[fill=black!30]
        \node[state,inner sep=0pt,minimum size=15pt] (A) [filled, label=a]                                  {};
        \node[state,inner sep=0pt,minimum size=15pt] (B) [above right=0.12cm and 0.5cm of A, label=b]       {};
        \node[state,inner sep=0pt,minimum size=15pt] (C) [below right=0.12cm and 0.5cm of A, label=below:c] {};
        \node[state,inner sep=0pt,minimum size=15pt] (D) [right=0.12cm and 0.5cm of B, label=d]             {};
        \node[state,inner sep=0pt,minimum size=15pt] (E) [right=0.12cm and 0.5cm of C, label=below:e]       {};
        \node[state,inner sep=0pt,minimum size=15pt] (F) [below right=0.12cm and 0.5cm of D, label=f]       {};

        \path
            (A) edge                 node [] {} (C)
            (B) edge                 node [] {} (C)
            (B) edge                 node [] {} (D)
            (B) edge                 node [] {} (E)
            (B) edge [bend right=15] node [] {} (F)
            (C) edge                 node [] {} (E)
            (D) edge                 node [] {} (E)
            (D) edge                 node [] {} (F)
            (E) edge                 node [] {} (F);
    \end{tikzpicture}
    \hspace*{15pt}
    \begin{tikzpicture}[align=center,semithick, every label/.append style={font=\scriptsize}]
        \tikzstyle{every state}=[draw=black,text=black]
        \tikzstyle{filled}=[fill=black!30]
        \node[state,inner sep=0pt,minimum size=15pt] (A) [filled, label=a]                                          {};
        \node[state,inner sep=0pt,minimum size=15pt] (B) [above right=0.12cm and 0.5cm of A, label=b]               {};
        \node[state,inner sep=0pt,minimum size=15pt] (C) [filled, below right=0.12cm and 0.5cm of A, label=below:c] {};
        \node[state,inner sep=0pt,minimum size=15pt] (D) [right=0.12cm and 0.5cm of B, label=d]                     {};
        \node[state,inner sep=0pt,minimum size=15pt] (E) [right=0.12cm and 0.5cm of C, label=below:e]               {};
        \node[state,inner sep=0pt,minimum size=15pt] (F) [below right=0.12cm and 0.5cm of D, label=f]               {};

        \path
            (A) edge                 node [] {} (C)
            (B) edge                 node [] {} (C)
            (B) edge                 node [] {} (D)
            (B) edge                 node [] {} (E)
            (B) edge [bend right=15] node [] {} (F)
            (C) edge                 node [] {} (E)
            (D) edge                 node [] {} (E)
            (D) edge                 node [] {} (F)
            (E) edge                 node [] {} (F);
    \end{tikzpicture}
\caption{Propagation of the behavior is stopped by a cluster of agents ($\theta = 1/3$).}
\label{fig:example-cluster}
\end{figure}

\section{Axioms for the LTL-SN system}
\label{axioms-system}

In this section the axioms of LTL-SN are presented. This formalism is based on reduction axioms, and this strategy allows us to ultimately reduce formulas into propositional logic. However before we can do that, we need to define a few abbreviations and enunciate some useful properties.

First, we will define the ``Majority'' abbreviation, just as is defined in \cite{DynamicEpistemicLogicsDiffusionPredictionSocialNetworks}. This formula determines whether an agent will adopt the behavior in the next position of the path, because it holds when a fraction of at least $\theta$ of its neighbors have already adopted the behavior.

\begin{defn}[Majority]
\label{majority}
An abbreviation representing ``majority'':
\[
\beta_{N(a) \geq \theta} \coloneqq \bigvee \limits _{\mathcal{G} \subseteq \mathcal{N} \subseteq \mathcal{A}: \frac{\lvert \mathcal{G} \rvert}{\lvert \mathcal{N} \rvert} \geq \theta} \left( \bigwedge \limits _{b \in \mathcal{N}} N_{ab} \land \bigwedge \limits _{b \not\in \mathcal{N}} \neg N_{ab} \land \bigwedge \limits _{b \in \mathcal{G}} \beta_b \right)
\]
\end{defn}

We also define an expansion that is equivalent to the Until operator, and demonstrate that both are indeed semantically equivalent. This expansion is a key part of the reduction axiomatization.

\begin{defn}[Until expansion]
\label{until-expansion} Let $\varphi U \psi$ be a formula. We define its Until expansion steps as
\begin{align*}
u^0(\varphi U \psi) &\coloneqq \psi \\
u^i(\varphi U \psi) &\coloneqq (\varphi \land X u^{i-1})
\end{align*}
Then, its Until expansion $EXP_U(\varphi U \psi)$ in a model with agent set $\mathcal{A}$ is defined as
\begin{equation*}
EXP_U(\varphi U \psi)  \coloneqq \bigvee \limits _{0 \leq i \leq \lvert \mathcal{A} \rvert} u^i(\varphi U \psi)
\end{equation*}
\end{defn}

Before we can show equivalence of the entire Until expansion, we demonstrate that its first $n$ steps are equivalent to the Until operator if and only if the right-hand side holds before the next $n$ positions (or vacuously never at all). After that, we also determine that the longest distance from the current position at which this can happen has to be less than $\lvert \mathcal{A} \rvert$ positions ahead. With both of these, we are able to prove equivalence of the complete Until expansion.

\begin{lem}[Until expansion step satisfaction]
\label{until-expansion-step-validity}
$\mathcal{M}, i \models u^n(\varphi U \psi)$ if and only if $\mathcal{M}, i + n \models \psi$ and $\mathcal{M}, i + j \models \varphi$ for all $0 \leq j < n$.
\end{lem}
\begin{proof}
We will abbreviate $u^n(\varphi U \psi)$ as $u^n$ for this proof. The proof goes by induction on $n$.

\textbf{Base case} ($n = 0$):

It is trivial to see that $\mathcal{M}, i \models u^0$ if and only if $\mathcal{M}, i \models \psi$, and there is no $j$ such that $0 \leq j < n$.

\textbf{Induction hypothesis} ($n = k$):

$\mathcal{M}, i \models u^k$ if and only if $\mathcal{M}, i + k \models \psi$, and $\mathcal{M}, i + j \models \varphi$ for all $0 \leq j < k$.

\textbf{Induction step} ($n = k + 1$):

Assume $\mathcal{M}, i \models u^{k + 1}$.

By definition we have that $\mathcal{M}, i \models \varphi \land X u^k$, if and only if $\mathcal{M}, i \models \varphi $ and $\mathcal{M}, i \models X u^k$. $\mathcal{M}, i \models X u^k$ if and only if $\mathcal{M}, i + 1 \models u^k$.

Without loss of generality, we can rename $i$ to $i + 1$ in the induction hypothesis, and then we have $\mathcal{M}, i + 1 \models u^k$ if and only if $\mathcal{M}, i + 1 + k \models \psi$, and $\mathcal{M}, i + 1+ j \models \varphi$ for all $0 \leq j < k$.

$\mathcal{M}, i + 1+ j \models \varphi$ for all $0 \leq j < k$ and $\mathcal{M}, i \models \varphi $ can be rewritten as $\mathcal{M}, i + j \models \varphi$ for all $0 \leq j < k + 1$.

Therefore we are able to conclude that $\mathcal{M}, i + (k + 1) \models \psi$ and $\mathcal{M}, i + j \models \varphi$ for all $0 \leq j < k + 1$, and that finishes our proof.
\end{proof}

\begin{lem}[Fixed-point]
\label{fixed-point}
Let $b = b_0, b_1, b_2, \ldots, b_n$ be a path. For some $i < \lvert \mathcal{A} \rvert$ a fixed-point $b_i = b_{i+1}$ is reached.
\end{lem}

\begin{proof}
This follows directly from the fact that $\mathcal{A}$ is finite and $b_i \subseteq b_{i+1}$ for all $i$ by definition. The slowest possible diffusion scenario is where $\lvert b_0 \rvert = 1$ and only one agent adopts the behavior per position on the path until the fixed-point is reached. It follows directly that the fixed-point will have to be reached for some $i < \lvert \mathcal{A} \rvert$.
\end{proof}

\begin{lem}[Until expansion soundness]
\label{until-expansion-soundness}
\[
\models \varphi U \psi \leftrightarrow EXP_U(\varphi U \psi)
\]
\end{lem}

\begin{proof}
By definition, $EXP_U(\varphi U \psi) = u^0(\varphi U \psi) \lor u^1(\varphi U \psi) \lor \ldots \lor u^{\lvert \mathcal{A} \rvert}(\varphi U \psi)$.

For any model $\mathcal{M}$ and position $j$, by lemma \ref{until-expansion-step-validity} we know that a term $u^i$ holds if and only if $\psi$ holds in position $j + i$. By lemma \ref{fixed-point} we know that if $\varphi U \psi$ holds, $\psi$ must hold at some position $i$ such that $0 \leq j + i \leq \lvert \mathcal{A} \rvert + j$. Therefore $\models \varphi U \psi$ if and only if $\models EXP_U(\varphi U \psi)$.
\end{proof}

\begin{table}[H]
\centering
\begin{tabular*}{\textwidth}{l@{\extracolsep{\fill}}ll}
\hline
\multicolumn{2}{l}{Network  axioms} \\ \hline
$\neg N_{aa}$ & Irreflexivity \\
$N_{ab} \leftrightarrow N_{ba}$ & Symmetry \\
$\bigvee \limits _{b \in \mathcal{A}} N_{ab}$ & Seriality \\
\hline
\multicolumn{2}{l}{Reduction axioms} \\ \hline
$X N_{ab} \leftrightarrow N_{ab}$ & Red.Ax.$X$.$N$ \\
$X \beta_a \leftrightarrow \beta_a \lor \beta_{N(a)} \geq \theta$ & Red.Ax.$X$.$\beta$ \\
$\varphi U \psi \leftrightarrow EXP_U(\varphi U \psi)$ & Red.Ax.$U$ \\
\hline
\multicolumn{2}{l}{LTL axioms} \\ \hline
\textit{All classical propositional tautologies} \\
$G(\varphi \rightarrow \psi) \rightarrow (G\varphi \rightarrow G\psi)$ & A1 \\
$\neg X \varphi \leftrightarrow X \neg \varphi$ & A2 \\
$X(\varphi \rightarrow \psi) \rightarrow (X \varphi \rightarrow X \psi)$ & A3 \\
$G(\varphi \rightarrow X \varphi) \rightarrow (\varphi \rightarrow G \varphi)$ & A4 \\
$(\varphi U \psi) \leftrightarrow \psi \lor (\varphi \land X(\varphi U \psi))$ & A5 \\
$(\varphi U \psi) \rightarrow F \psi$ & A6 \\
$\varphi U (\psi_1 \lor \psi_2) \leftrightarrow \varphi U \psi_1 \lor \varphi U \psi_2$ & A7 \\
$X (\varphi \land \psi) \leftrightarrow X \varphi \land X \psi$ & A8 \\
$\displaystyle{\seq{\varphi \qquad \varphi \rightarrow \psi}{\psi}}$ & MP \\
$\displaystyle{\seq{\vdash \varphi}{\vdash G \varphi}}$ & NG \\
$\displaystyle{\seq{\vdash \varphi}{\vdash X \varphi}}$ & NX \\
\end{tabular*}
\caption{LTL-SN axioms.}
\end{table}

LTL axioms are as seen on \cite{MadsDamLecture}. Notice that for completeness of this particular class of models, only axioms $A2$, $A3$, $A8$ and the rules of inference $MP$ and $NX$ were used. The others were kept merely for reference.

\begin{prop}[Next rule] \label{X} The following inference rule is admissible:

$\displaystyle{\seq{\vdash \phi \leftrightarrow \psi}{\vdash X\phi \leftrightarrow X\psi}}$

\end{prop}

\begin{prop}[Replacement of equivalents] \label{EQ} If $\vdash \phi_1 \leftrightarrow \phi_2$, then \\ $\vdash [\phi_1/\psi]\chi \leftrightarrow [\phi_2/\psi]\chi$.

\end{prop}

The proof of proposition \ref{X} follows straightforward from axiom A3 and inference rules NX and MP. The proof of proposition \ref{EQ} follows by induction on the construction of $\chi$. These proofs are rather standard in LTL literature \cite{burgess1982axioms}.

\section{Soundness and completeness}
\label{soundness-and-completeness}
Since the formulas of LTL-SN are evaluated over conventional LTL paths (or ``flows of time'' \cite{10.2307/40180020,doi:https://doi.org/10.1002/9781405164801.ch10}), all standard LTL axioms are also sound in this logic.


\begin{lem}[Soundness]
Let $\mathcal{M}$ be an arbitrary LTL-SN model and $B \in \mathcal{B}$ a behavior set.
\end{lem}

\begin{proof}
Consider each case:
\begin{itemize}
    \item All network axioms follow directly from definition \ref{defn-model};
    \item $\mathcal{M}, i \models X N_{ab} \leftrightarrow N_{ab}$ follows directly from the fact that $\leq$ never alters the network structure;
    \item $\mathcal{M}, i \models \varphi U \psi \leftrightarrow EXP_U(\varphi U \psi)$ follows directly from lemma \ref{until-expansion-soundness};
    \item For $\mathcal{M}, i \models X \beta_a \leftrightarrow \beta_a \lor \beta_{N(a)} \geq \theta$, consider $B = b_i$ and a $B'$ such that $B \leq B'$.
    
    $\mathcal{M}, i \models X \beta_a $ iff $\mathcal{M}, i + 1 \models \beta_a $ iff $a \in B' = B \cup \left\{ b \in \mathcal{A} : \frac{\lvert N(b) \cap B \rvert}{\lvert N(b) \rvert} > \theta \right\}$ iff $\mathcal{M}, i \models \beta_a$ or $a \in \left\{ b \in \mathcal{A} : \frac{\lvert N(b) \cap B \rvert}{\lvert N(b) \rvert} > \theta \right\}$.
    
    It can be shown that the large disjunct in definition \ref{majority} is satisfied iff $a \in \left\{ b \in \mathcal{A} : \frac{\lvert N(b) \cap B \rvert}{\lvert N(b) \rvert} > \theta \right\}$ just as in \cite{DynamicEpistemicLogicsDiffusionPredictionSocialNetworks}.
    
    Hence $\mathcal{M}, i \models X \beta_a$ iff $\mathcal{M}, i \models \beta_a$ or $\mathcal{M}, i \models \beta_{N(A) \geq \theta}$.
\end{itemize}
\end{proof}

\subsection{Completeness}
\label{completeness}
 The strategy for this completeness proof is a two-step translation of LTL-SN formulas into propositional formulas. The first step replaces Until operators, and the second step replaces Next operators. The proofs of equivalence of each step goes by induction.

\begin{defn}[Until Translation]\label{tu}
\begin{align*}
    t_u(N_{ab}) &= N_{ab} \\
    t_u(\beta_a) &= \beta_a \\
    t_u(\varphi \land \psi) &= t_u(\varphi) \land t_u(\psi) \\
    t_u(\neg \varphi) &= \neg t_u(\varphi) \\
    t_u(X \varphi) &= X t_u(\varphi) \\
    t_u(\varphi U \psi) &= t_u(EXP_U(\varphi U \psi))
\end{align*}
\end{defn}

Next lemma assures that every formula is equivalent to another formula without any occurrence of the $until$ operator.

\begin{lem}
\[
\vdash \varphi \leftrightarrow t_u(\varphi)
\]
\end{lem}

\begin{proof}
Proof by induction on the length $|\varphi| =n$ of  the formula $\varphi$.

\noindent \textbf{Base case} ($n = 1$): $\varphi = N_{ab}$ or $\varphi = \beta_{a}$

This is straightforward from definition \ref{tu}.

\noindent \textbf{Induction hypothesis} ($|\varphi| \leq n$):

We have four cases:
\begin{enumerate}
    \item $\varphi = \neg \psi$: we have that $t_u(\neg \psi) = \neg t_u(\psi)$. 
    
    By the I. H., $\vdash \psi \leftrightarrow t_u(\psi)$ and\\
    $\vdash \neg \psi \leftrightarrow \neg t_u(\psi)$. Thus $\vdash \varphi \leftrightarrow  t_u(\neg \psi)$.
    
     \item $\varphi = \varphi_1 \land \varphi_2$: we have that $t_u(\varphi_1 \land \varphi_2) =  t_u(\varphi_1) \land t_u(\varphi_2)$. 
    
    By the I. H., $\vdash \varphi_1 \leftrightarrow t_u(\varphi_1)$ and $\vdash \varphi_2 \leftrightarrow t_u(\varphi_2)$.\\
     Thus $\vdash \varphi_1 \land \varphi_2 \leftrightarrow  t_u(\varphi_1) \land t_u(\varphi_2)$ and \\
     $\vdash \varphi_1 \land \varphi_2 \leftrightarrow  t_u(\varphi_1 \land \varphi_2)$.
     
      \item $\varphi = X \psi$: we have that $t_u(X \psi) = X t_u(\psi)$. 
    
    By the I. H., $\vdash \psi \leftrightarrow t_u(\psi)$ and\\
    Using using Next rule (proposition \ref{X}), $\vdash X \psi \leftrightarrow X t_u(\psi)$. \\ Thus $\vdash \varphi \leftrightarrow  t_u(X \psi)$.
    
     \item $\varphi = \varphi_1 U \varphi_2$: we have that $t_u(\varphi_1 U \varphi_2) =  t_u(EXP_U(\varphi_1 U \varphi_2))$. 
     
     \textbf{Claim 1:} $\vdash t_u(EXP_U(\varphi_1 U \varphi_2)) \leftrightarrow EXP_U(t_u(\varphi_1) U t_u(\varphi_2))$
     
      \textbf{Proof:} $EXP_U(t_u(\varphi_1) U t_u(\varphi_2)) = \bigvee \limits _{0 \leq i \leq \lvert \mathcal{A} \rvert} u^i(t_u(\varphi_1) U t_(\varphi_2))$
      
      = $u^0 \lor u^1 \lor ... = t_u(\varphi_2) \lor (t(\varphi_1) \land X t_u(\varphi_2)) \lor ...$

    Using 1., 2. and 3., we can bring the translation operator outside yielding  $\vdash t_u(EXP_U(\varphi_1 U \varphi_2)) \leftrightarrow EXP_U(t_u(\varphi_1) U t_u(\varphi_2))$.
    
    Returning to the proof of 4., we have $t_u(\varphi_1 U \varphi_2) = t_u(EXP_U(\varphi_1 U \varphi_2)) \\ = EXP_U(t_u(\varphi_1) U t_u(\varphi_2))$ by Claim 1. Using axiom Red.Ax.$U$, we have $\vdash EXP_U(t_u(\varphi_1) U t_u(\varphi_2)) \leftrightarrow (t_u (\varphi_1) U t_u(\varphi_2))$
    
    By the I. H., $\vdash \varphi_1 \leftrightarrow t_u(\varphi_1)$ and $\vdash \varphi_2 \leftrightarrow t_u(\varphi_2)$. And so\\
    $\vdash EXP_U(t_u(\varphi_1) U t_u(\varphi_2)) \leftrightarrow  (\varphi_1 U \varphi_2)$.
    
    Using Claim 1,we obtain $\vdash t_u(EXP_U(\varphi_1 U \varphi_2)) \leftrightarrow  (\varphi_1 U \varphi_2)$. Thus, \\ $\vdash t_u(\varphi_1 U \varphi_2) \varphi \leftrightarrow (\varphi_1 U \varphi_2) $.
    
\end{enumerate}
\end{proof}

\begin{defn}[Translation]
Using the reduction axioms it is possible to translate an LTL-SN formula with no Until occurrences into propositional logic.

\begin{align*}
    t(N_{ab}) &= N_{ab} \\
    t(\beta_a) &= \beta_a \\
    t(\varphi \land \psi) &= t(\varphi) \land t(\psi) \\
    t(\neg \varphi) &= \neg t(\varphi) \\
    t(X N_{ab}) &= N_{ab} \\
    t(X \beta_a) &= t(\beta_a \lor \beta_{N(a)} \geq \theta) \\
    t(X (\varphi \land \psi)) &= t(X \varphi \land X \psi) \\
    t(X \neg \varphi) &= t(\neg X \varphi) \\
    t(X X \varphi) &= t(X t(X \varphi)) \\
\end{align*}
Translation inspired by the ones from \cite{10.5555/1535423,ReductionAxiomsForEpistemicActions}.
\end{defn}

\begin{defn}[Cost measure]
\begin{align*}
    c(\beta_a) &= 1 \\
    c(N_{ab}) &= 1 \\
    c(\neg \varphi) &= 1 + c(\varphi) \\
    c(\varphi_1 \land \varphi_2) &= 1 + max(c(\varphi_1), c(\varphi_2)) \\
    c(X \beta_a) &= 4 + 2 \cdot \lvert \mathcal{A} \rvert^2 \\
    c(X \psi) &= 2 \cdot c(\psi) \\
\end{align*}
where $X \psi$ matches any formula besides those of form $X \beta_a$.
\end{defn}

\begin{lem}[Cost of majority abbreviation]
\label{cost-majority-abbreviation}
The cost of the majority abbreviation $c(\beta_{N(a)} \geq \theta)$ in a model with agent set $\mathcal{A}$ is $2 \cdot \lvert \mathcal{A} \rvert^2$.
\end{lem}

\begin{proof}
The outer disjunction ranges over two groups, $\mathcal{G} \subseteq \mathcal{N} \subseteq \mathcal{A}$, therefore in the worst case there will be as many terms as $\lvert \mathcal{A} \rvert^2$ for each possible combinations of both. Now analyzing each term, the first two conjunctions contribute at most $\lvert \mathcal{A} \rvert$ propositions, and the last conjunction another $\lvert \mathcal{A} \rvert$. Therefore each term contributes, in the worst case, a cost of $2 \cdot \lvert \mathcal{A} \rvert$. Therefore, the total cost is $2 \cdot \lvert \mathcal{A} \rvert^2$.
\end{proof}

\begin{lem}[Translation costs]
\label{translation-costs}

For all $\varphi$, $\psi$:

\begin{alignat*}{3}
    &1. \  c(\varphi) &&\geq c(\psi) \  \text{if} \  \psi \in Sub(\varphi) \\
    &2. \  c(X N_{ab}) &&> c(N_{ab}) \\
    &3. \  c(X \beta_a) &&> c(\beta_a \lor \beta_{N(a)} \geq \theta) \\
    &4. \  c(X(\varphi \land \psi)) &&> c(X \varphi) \land c(X \psi) \\
    &5. \  c(X \neg \varphi) &&> \neg c(X \varphi) \\
    &6. \  c(X X \varphi) &&> c(X t(X \varphi))
\end{alignat*}
\end{lem}

\begin{proof}
We analyze case by case:

Case $c(\varphi) \geq c(\phi)$ if $\phi \in Sub(\varphi)$:

By induction on $\varphi$.

\textbf{Base case:} if $\varphi$ is $\beta_a$ or $N_{ab}$, its complexity is $1$ and it is its only subformula.

\textbf{Induction hypothesis:} $c(\varphi) \geq c(\phi)$ if $\phi \in Sub(\varphi)$.

\textbf{Induction step:}
\begin{itemize}
    \item Negation ($\neg \varphi$):
    $\phi$ is a subformula of $\neg \varphi$, therefore $\phi$ is either $\neg \varphi$ itself or a subformula of $\varphi$. In the first case it follows directly that $c(\neg \varphi) \geq c(\phi)$. In the second case, we have that $c(\neg \varphi) = 1 + c(\varphi)$, and since $\phi$ is a subformula of $\varphi$, it follow directly by the I. H. that $c(\neg \varphi) \geq c(\phi)$.
    \item Conjunction ($\varphi \land \varphi'$):
    $\phi$ is a subformula of $\varphi \land \varphi'$, therefore $\phi$ is either $\varphi \land \varphi'$ itself or a subformula of $\varphi$ or $\varphi'$. In the first case it follows directly that $c(\varphi \land \varphi') \geq c(\phi)$. In the second case, we have that $c(\varphi \land \varphi') = 1 + max(c(\varphi), c(\varphi'))$, and since $\phi$ is a subformula of either $\varphi$ or $\varphi'$, it follow directly by the I. H. that $c(\varphi \land \varphi') \geq c(\phi)$.
\end{itemize}

Case $c(X N_{ab}) > c(N_{ab})$:
\begin{align*}
    c(X N_{ab}) &= 2 \cdot c(N_{ab}) = 2 \\
    &\text{and} \\
    c(N_{ab}) &= 1 \\
\end{align*}

Case $c(X \beta_a) > c(\beta_a \lor \beta_{N(a)} \geq \theta)$:
\begin{align*}
    c(X \beta_a) &= 4 + 2 \cdot \lvert \mathcal{A} \rvert^2 \\
    &\text{and} \\
    c(\beta_a \lor \beta_{N(a)} \geq \theta) &= c(\neg (\neg \beta_a \land \neg \beta_{N(a)} \geq \theta)) \\
    &= 1 + c(\neg \beta_a \land \neg \beta_{N(a)} \geq \theta) \\
    &= 2 + max(c(\neg \beta_a), c(\neg \beta_{N(a)} \geq \theta)) \\
    &= 2 + max(1 + c(\beta_a), 1 + c(\beta_{N(a)} \geq \theta)) \\
    &= 2 + max(2, 1 + c(\beta_{N(a)} \geq \theta)) \\
    &= 2 + max(2, 1 + 2 \cdot \lvert \mathcal{A} \rvert^2) \\
    &= 3 + 2 \cdot \lvert \mathcal{A} \rvert^2 \\
\end{align*}
$3 + 2 \cdot \lvert \mathcal{A} \rvert^2$ is less than $4 + 2 \cdot \lvert \mathcal{A} \rvert^2$.

Case $c(X (\varphi \land \psi)) > c(X \varphi \land X \psi)$:
\begin{align*}
    c(X (\varphi \land \psi)) &= 2 \cdot c(\varphi \land \psi) = 2 \cdot (1 + max(c(\varphi), c(\psi))) \\
    &= 2 + 2 \cdot max(c(\varphi), c(\psi)) \\
    &\text{and} \\
    c(X \varphi \land X \psi) &= 1 + max(c(X \varphi), c(X \psi)) = 1 + max(2 \cdot c(\varphi), 2 \cdot c(\psi)) \\
    &= 1 + 2 \cdot max(c(\varphi), c(\psi)) \\
\end{align*}
$1 + 2 \cdot max(c(\varphi), c(\psi))$ is less than $2 + 2 \cdot max(c(\varphi), c(\psi))$.

Case $c(X \neg \varphi) > c(\neg X \varphi)$:
\begin{align*}
    c(X \neg \varphi) &= 2 \cdot c(\neg \varphi) = 2 \cdot (1 + c(\varphi)) \\
    &= 2 + 2 \cdot c(\varphi) \\
    &\text{and} \\
    c(\neg X \varphi) &= 1 + c(X \varphi) = 1 + 2\cdot c(\varphi) \\
\end{align*}
$1 + 2\cdot c(\varphi)$ is less than $2 \cdot (1 + c(\varphi))$.

Case $c(X X \varphi) > c(X t(X \varphi))$:

\begin{align*}
    c(X X \varphi) &= 2 \cdot c(X \varphi) \\
    &\text{and} \\
    c(X t(X \varphi)) &= 2 \cdot c(t(X \varphi)) \\
\end{align*}

Now we show $c(X \varphi) > c(t(X \varphi))$ by induction on the length of $\varphi$.

\textbf{Base case:} if $\varphi$ has length $1$, then it is either $\beta_a$ or $N_{ab}$:
\begin{align*}
    c(X \varphi) &= 2 \cdot 1 = 2 \\
    &\text{and} \\
    c(t(X \varphi)) &= c(N_{ab}) = 1 \\
    &\text{or} \\
    c(t(X \varphi)) &= c(\beta_a) = 1 \\
\end{align*}

\textbf{Induction hypothesis:} $c(X \varphi) > c(t(X \varphi))$ for $\lvert \varphi \rvert \leq n$.

\textbf{Induction step:}

Case $\varphi = \neg \phi$:
\begin{align*}
    c(X \varphi) &= c(X \neg \phi) = c(\neg X \phi) = 1 + c(X \phi) \\
    &\text{and} \\
    c(t(X \varphi)) &= c(t(X \neg \phi)) = c(t(\neg X \phi)) = c(\neg t(X \phi)) \\
    &= 1 + c(t(X \phi)) \\
\end{align*}
By I. H., we can conclude $1 + c(X \phi) > 1 + c(t(X \phi))$, therefore $c(X \varphi) > c(t(X \varphi))$.

Case $\varphi = \phi_1 \land \phi_2$:
\begin{align*}
    c(X \varphi) &= c(X(\phi_1 \land \phi_2)) = 2 \cdot c(\phi_1 \land \phi_2) \\
    &= 2 \cdot (1 + max(c(\phi_1), c(\phi_2))) = 2 + 2 \cdot max(c(X \phi_1), c(X \phi_2)) \\
    &\text{and} \\
    c(t(X \varphi)) &= c(t(X (\phi_1 \land \phi_2))) = c(t(X \phi_1 \land X \phi_2)) \\
    &= c(t(X \phi_1) \land t(X \phi_2)) = 1 + max(c(t(X \phi_1)), c(t(X \phi_2))) \\
\end{align*}
By I. H., we can conclude \\ $2 + 2 \cdot max(c(X \phi_1), c(X \phi_2)) > 1 + max(c(t(X \phi_1)), c(t(X \phi_2)))$, therefore $c(X \varphi) > c(t(X \varphi))$.

Case $\varphi = X \phi$:
\begin{align*}
    c(X \varphi) &= c(X X \phi) = 2 \cdot c(X \phi) \\
    &\text{and} \\
    c(t(X \varphi)) &= c(t(X X \phi)) = c(t(X t(X \phi))) = \Box
\end{align*}
By I. H., we know $c(X t(X \phi)) > \Box$. $c(X t(X \phi)) = 2 \cdot c(t(X \phi))$. By I. H. we know $2 \cdot c(X \phi) > 2 \cdot c(t(X \phi)) > \Box$, therefore $c(X \varphi) > c(t(X \varphi))$.
\end{proof}

\begin{lem}[Translation equivalence]
\label{translation-equivalence}
\[
\vdash \varphi \leftrightarrow t(\varphi)
\]
\end{lem}
\begin{proof}
By induction on $c(\varphi)$.

\textbf{Base case:} $\varphi$ is either $\beta_a$ or $N_{ab}$, translation keeps them unchanged and therefore $\vdash \varphi \leftrightarrow t(\varphi)$.

\textbf{Induction hypothesis:} For all $\varphi$ such that $c(\varphi) < n : \  \vdash \varphi \leftrightarrow t(\varphi)$.

\textbf{Induction step:}
Case for negation and conjunction:
straightforward from lemma \ref{translation-costs} item 1.

Case $X N_{ab}$:
This case follows from the $Red.Ax.X.N$ axiom, item 2 of lemma \ref{translation-costs} and the I. H..

Case $X \beta_a$:
This case follows from the $Red.Ax.X.\beta$ axiom, item 3 of lemma \ref{translation-costs} and the I. H..

Case $X (\varphi \land \psi)$:
This case follows from the $A8$ axiom, item 4 of lemma \ref{translation-costs} and the I. H..

Case $X \neg \varphi$:
This case follows from the $A2$ axiom, item 5 of lemma \ref{translation-costs} and the I. H..

Case $X X \varphi$:
By I. H. we know $\vdash X \varphi \leftrightarrow t(X \varphi)$. Therefore we can replace into $X t(X \varphi)$ by proposition \ref{EQ}, and it follows from item 6 of lemma \ref{translation-costs}.
\end{proof}

\begin{thm}[Completeness]
For every $\varphi \in \mathcal{LTL-SN}$:
\[
\models \varphi \  \text{implies} \  \vdash \varphi
\]
\end{thm}
\begin{proof}
Suppose $\models \varphi$. By lemma \ref{translation-equivalence} and since the proof system is sound, we have that $\models t(\varphi)$. The formula $t(\varphi)$ is propositional, therefore $\mathcal{PROP} \vdash t(\varphi)$ by completeness of propositional logic. We also have that $\mathcal{LTL-SN} \vdash t(\varphi)$ since $\mathcal{PROP}$ is a subsystem of LTL-SN. Again by lemma \ref{translation-equivalence} we have that $\mathcal{LTL-SN} \vdash \varphi$.
\end{proof}

\section{Model Checking}
\label{model-checking}

In this section we analyse the computational complexity of the model checking problem.

\begin{defn} The {\it model checking problem} consist of, given a formula $\phi$ and a finite  LTL-SN model ${\cal M} = ( \mathcal{A}, N , \theta, I)$, determining the set ${\cal S}(\phi) = \{ i \mid b_i \in b~ \mathrm{and} ~{\cal M} , i \models \phi \}$

\end{defn}

Next, we present the model checking algorithm for LTL-SN. Let $label(i)$ denote the set of sub-formulas that hold at path position $i$. We start by initializing $label(i) := \{\beta_a \mid a \in b_i \} \cap \{N_{ab} \mid b \in N_a \}$ for all path positions $i$.

\begin{algorithm}[H]
		
\caption{procedure Check($\phi$)} \label{procCheck}
\begin{algorithmic}
\WHILE {$|\phi|\geq 1$}
    \IF {$\phi= (\lnot\phi_1)$} \STATE $Check(\phi_1)$; $CheckNOT(\phi_1)$
    \ELSIF {$\phi= (\phi_1 \land \phi_2)$}
     \STATE $Check(\phi_1)$; $Check(\phi_2)$; $CheckAND(\phi_1,\phi_2)$
    \ELSIF {$\phi= X \phi_1$}
     \STATE $Check(\phi_1)$; $CheckX(\phi_1)$
     \ELSIF {$\phi= \phi_1 U \phi_2$}
     \STATE $Check(\phi_1)$;$Check(\phi_2)$; $CheckU(\phi_1,\phi_2)$
    \ENDIF
 \ENDWHILE
\end{algorithmic}
\end{algorithm}

\begin{algorithm}[H]
\caption{procedure CheckNOT($\phi$)} \label{procCheckNOT}
\begin{algorithmic}
\FORALL {$i$}
    \IF {$\phi\not\in label(i)$}
     \STATE $label(i):=label(i)\cup \{(\lnot\phi)\}$
    \ENDIF
 \ENDFOR
\end{algorithmic}
\end{algorithm}

\begin{algorithm}[H]
\caption{procedure CheckAND($(\phi_1 \land \phi_2$)}
\label{procCheckIMP}
\begin{algorithmic}
\FORALL {$i$}
    \IF {$\phi_1 \in label(i)$\textbf{ and }$\phi_2\in label(i)$}
     \STATE $label(i):=label(i)\cup \{(\phi_1 \land \phi_2)\}$
    \ENDIF
 \ENDFOR
\end{algorithmic}
\end{algorithm}

\begin{algorithm}[H]
\caption{procedure CheckX($\phi_1$)} \label{procCheckDIAM}
\begin{algorithmic}
\STATE $T:=\{i~|~\phi_1\in label(i)\}$
\WHILE {$T\not=\emptyset$}
    \STATE \textbf{choose $i \in T$}
    \STATE $T:=T\setminus\{i\}$
    \IF {$ X \phi_1 \not \in label(i-1)$ } 
            \STATE $label(i-1):=label(i-1)\cup \{ X \phi_1 \}$
    \ENDIF
 \ENDWHILE
\end{algorithmic}
\end{algorithm}

\begin{algorithm}[H]
\caption{procedure CheckEU($\alpha_1$,$\alpha_2$)}
\label{procCheckEU}
\begin{algorithmic}
\STATE $T:=\{i~|~\phi_2\in label(i)\}$
\FORALL {$i\in T$}
\STATE $label(i):=label(i)\cup\{E(\alpha_1\mathcal{U}\alpha_2)\}$ \ENDFOR
\WHILE {$T\neq\emptyset$}
    \STATE \textbf{choose maximum}$~i\in T$
    \STATE $T:=T\backslash\{i\}$
    \WHILE {$\phi_1 \in label(i-1)$}
    \IF {$E(\phi_1\mathcal{U}\phi_2)\not\in label(i-1)$ }
     \STATE $label(i-1):=label(i-1)\cup \{E(\phi_1\mathcal{U}\phi_2)\}$
     \STATE $i:=i-1$
     \STATE $T:=T\backslash\{i\}$
    \ENDIF
 \ENDWHILE
 \ENDWHILE
\end{algorithmic}
\end{algorithm}

\begin{thm} Given a LTL-SN model ${\cal M} = ( \mathcal{A}, N , \theta, I)$ and formula $\phi$. The computational complexity of the model checking problem is $O(\mid \phi \mid \times \mid \mathcal{A} \mid)$, i.e., linear in the size of the formula times the size of the set of agents.

\begin{proof} By lemma \ref{path-uniqueness}, there exists a unique path such that $b_0=I$. By lemma \ref{fixed-point}, this path $b_0,..., b_n$ reaches a fixed-point such that $n < |\mathcal{A}|$. 

The algorithm Check($\phi$) is called once for each sub-formula of $\phi$ which is $O(|\phi|)$ and each time it activates the algorithms CheckNOT, CheckAND, CheckX and CheckU. The algorithm CheckX, in the worse case, has to visit all $b_i$, which is $O(|\mathcal{A}|)$ The CheckU, starts with the greatest $i$ where $\phi_2$ holds, and then it goes labelling states smaller than i, where $\phi_1$ holds, with $\phi_1 U \phi_2$. In worse case, it has to search the whole model which is also $O(|\mathcal{A}|)$.  The algorithms CheckNOT, CheckAND take constant time. Thus, the complexity of Check($\phi$) is $O(|\mathcal{A}| \times |\phi|)$. In order to o build the set $S(\phi)$ we only need to search $label(i)$, for all $0 \leq i < |\mathcal{A}| $, and check if $\phi \in label(i)$.
\end{proof}
\end{thm}

\section{Conclusion and future works}
\label{conclusion-future-works}
The paper presented a class of models for the study of social networks, based on threshold models as seen across the literature. With this, we have developed a logic with the same operators as LTL, but for a restricted class of models focused specifically on the study of social networks, proved its soundness and its completeness via a translation argument, and briefly discussed model checking complexity. With this framework one is able to capitalize on existing model checking solutions for LTL, which we believe is a good incentive to pursue this line of work.

Given the previous conclusion, a natural future work is to explore these model checkers and analyze the evolution of some network instances and compare them against data from real social networks.

In regard to the logic itself, we also believe it would be interesting to pursue a branching paths model, where agents may or may not adopt a behavior after each step. This would further leverage the expressive power of LTL, however a sound and complete axiomatization will require some more work, as a reduction into propositional logic may not work anymore.

\subsection{Acknowledgements}
We would like to thank the research agencies CNPq, CAPES and FAPERJ for their support.

\bibliography{bibliography}

\begin{thebibliography}{10}
\providecommand{\url}[1]{\texttt{#1}}
\providecommand{\urlprefix}{URL }
\providecommand{\doi}[1]{https://doi.org/#1}

\bibitem{10.1007/978-3-642-24829-0_20}
Apt, K.R., Markakis, E.: Diffusion in social networks with competing products.
  In: Persiano, G. (ed.) Algorithmic Game Theory. pp. 212--223. Springer Berlin
  Heidelberg, Berlin, Heidelberg (2011)

\bibitem{DynamicEpistemicLogicsDiffusionPredictionSocialNetworks}
Baltag, A., Christoff, Z., Kræmmer~Rendsvig, R., Smets, S.: Dynamic epistemic
  logics of diffusion and prediction in social networks. Studia Logica  (07
  2018). \doi{10.1007/s11225-018-9804-x}

\bibitem{burgess1982axioms}
Burgess, J.P., et~al.: Axioms for tense logic. i.``since''and``until''. Notre
  Dame Journal of Formal Logic  \textbf{23}(4),  367--374 (1982)

\bibitem{ReflectingSocialInfluenceNetworks}
Christoff, Z., Hansen, J.U., Proietti, C.: Reflecting on social influence in
  networks. J. of Logic, Lang. and Inf.  \textbf{25}(3–4),  299–333 (Dec
  2016). \doi{10.1007/s10849-016-9242-y},
  \url{https://doi.org/10.1007/s10849-016-9242-y}

\bibitem{MadsDamLecture}
Dam, M.: Lecture notes in temporal logic (1 2021),
  \url{http://www.csc.kth.se/~mfd/Courses/Temporal_logic/course_description.php}

\bibitem{ReachingConsensus}
DeGroot, M.H.: Reaching a consensus. Journal of the American Statistical
  Association  \textbf{69}(345),  118--121 (1974),
  \url{http://www.jstor.org/stable/2285509}

\bibitem{10.5555/1535423}
Ditmarsch, H.v., van~der Hoek, W., Kooi, B.: Dynamic Epistemic Logic. Springer
  Publishing Company, Incorporated, 1st edn. (2007)

\bibitem{10.2307/40180020}
Finger, M., Gabbay, D.M.: Adding a temporal dimension to a logic system.
  Journal of Logic, Language, and Information  \textbf{1}(3),  203--233 (1992),
  \url{http://www.jstor.org/stable/40180020}

\bibitem{ReductionAxiomsForEpistemicActions}
Kooi, B., Benthem, J.: Reduction axioms for epistemic actions. Preliminary
  Proceedings of AiML-2004  (01 2004)

\bibitem{WhenRationalDisagreementImpossible}
Lehrer, K.: When rational disagreement is impossible. Noûs  \textbf{10}(3),
  327--332 (1976), \url{http://www.jstor.org/stable/2214612}

\bibitem{DBLP:journals/corr/abs-1708-01477}
Rendsvig, R.K.: Diffusion, influence and best-response dynamics in networks: An
  action model approach. CoRR  \textbf{abs/1708.01477} (2017),
  \url{http://arxiv.org/abs/1708.01477}

\bibitem{doi:https://doi.org/10.1002/9781405164801.ch10}
Venema, Y.: Temporal Logic, chap.~10, pp. 203--223. John Wiley and Sons, Ltd
  (2017). \doi{https://doi.org/10.1002/9781405164801.ch10},
  \url{https://onlinelibrary.wiley.com/doi/abs/10.1002/9781405164801.ch10}

\bibitem{ZHEN2011275}
Zhen, L., Seligman, J.: A logical model of the dynamics of peer pressure.
  Electronic Notes in Theoretical Computer Science  \textbf{278},  275 -- 288
  (2011). \doi{https://doi.org/10.1016/j.entcs.2011.10.021},
  \url{http://www.sciencedirect.com/science/article/pii/S1571066111001496},
  proceedings of the 7th Workshop on Methods for Modalities (M4M’2011) and
  the 4th Workshop on Logical Aspects of Multi-Agent Systems (LAMAS’2011)

\end{thebibliography}

\end{document}